%% file: db_aeu11_arxiv.tex
\documentclass[conference]{IEEEtran-v17}
\pdfoutput=1 

\usepackage{xspace}
\usepackage[nosort]{cite}        
\usepackage{url} 
\usepackage[cmex10]{amsmath}
\usepackage{bbm}
\usepackage{graphicx}
\usepackage{paralist}
\usepackage{fancyref}
\usepackage[stretch=16,shrink=16,step=4]{microtype}
\usepackage{pdfsync}
\newlength{\figwidth}
\makeatletter
\def\@IEEEinterspaceratioM{0.265}
\def\@IEEEinterspaceMINratioM{0.1651}
\def\@IEEEinterspaceMAXratioM{0.38}

\def\@IEEEinterspaceratioB{0.31}
\def\@IEEEinterspaceMINratioB{0.19}
\def\@IEEEinterspaceMAXratioB{0.38}
\@IEEEtunefonts
\makeatother

\input{vmr-symbols-vecbold.tex}
\input{standard-macros.tex}

\input{defs}

\begin{document}
\IEEEoverridecommandlockouts

\title{High-SNR Capacity of Wireless Communication Channels in the Noncoherent Setting: A Primer}



\author{
\IEEEauthorblockN{Giuseppe Durisi}
\IEEEauthorblockA{Department of Signals and Systems\\
 Chalmers University of Technology, 41296 Gothenburg, Sweden\\
E-mail: \{durisi\}@chalmers.se\\}     
\IEEEauthorblockN{Helmut B\"olcskei}
\IEEEauthorblockA{Department of Information Technology and Electrical Engineering\\
 ETH Zurich, 8092 Zurich, Switzerland\\
E-mail: \{boelcskei\}@nari.ee.ethz.ch\\}
}

\maketitle

\begin{abstract}
 	This paper, mostly tutorial in nature, deals with the problem of characterizing the capacity of fading channels in the high signal-to-noise ratio (SNR) regime.
	We focus on the practically relevant \emph{noncoherent setting}, where neither  transmitter nor  receiver know the channel realizations, but both are aware of the channel law.
	We present, in an intuitive and accessible form, two tools, first proposed by Lapidoth \& Moser (2003), of fundamental importance to high-SNR capacity analysis:
	the \emph{duality approach} and the \emph{escape-to-infinity property} of capacity-achieving distributions.
	Furthermore, we apply these tools to refine some of the results that appeared previously in the literature and to simplify the corresponding proofs.
\end{abstract}
\IEEEpeerreviewmaketitle
\section{Introduction}
%
%

Most wireless communication systems operate in the \emph{noncoherent setting} where neither transmitter nor receiver have \emph{a priori} information on the realization of the underlying fading channel. 
As channel state information is typically acquired by allocating transmission time and/or bandwidth to channel estimation
(a typical example is the use of pilot symbols~\cite{tong04-11a}),
a problem of significant practical relevance is to determine the optimal amount of resources to be used for this task.
This problem can be addressed in a fundamental fashion by determining the Shannon capacity (i.e., the ultimate limit on the rate of reliable communication~\cite{cover06-a}) in the noncoherent setting.
Unfortunately, corresponding analytical results are exceedingly difficult to obtain, even for simple channel models~\cite{abou-faycal01-05a}; nevertheless, significant progress has been made during the past few years by studying the capacity behavior in the asymptotic regimes of high and low signal-to-noise ratio (SNR).
Throughout this paper, we shall deal exclusively with the high-SNR regime.  
The capacity behavior at high SNR turns out to be very sensitive to the channel model used~\cite{lapidoth03-10a,lapidoth05-07a,durisi11-03b}.
In this paper, we shall focus on a channel model---the \emph{correlated block-fading} model~\cite{liang04-09a,morgenshtern10-06a}---that is simple and yet rich enough to illustrate some of the possible asymptotic dependencies of capacity on SNR, namely, logarithmic with different \emph{pre-log} factors~\cite{liang04-09a,lapidoth05-07a,zheng02-02a,hochwald00-03a}, or double-logarithmic~\cite{lapidoth03-10a}. 
The aim of this tutorial paper is two-fold:
\begin{itemize}
	\item We present, in an intuitive and accessible manner, two tools that turn out to be exceedingly useful in the characterization of capacity at high SNR: the \emph{duality approach} and the \emph{escape-to-infinity property} of  capacity-achieving distributions.              
	These tools were first introduced in~\cite{lapidoth03-10a}.                   
	\item We use these tools to refine a result that appeared previously in~\cite{liang04-09a} and to provide an alternative and much simpler proof of a result in~\cite{zheng02-02a,hochwald00-03a}.      
	Furthermore, we develop insights into the use of duality by exploiting the geometry of the correlated block-fading model.
\end{itemize}

%

\subsection*{Notation} 
Uppercase boldface letters denote matrices and lowercase boldface letters
designate vectors.
Uppercase sans-serif letters (e.g., \inpprobmeas) denote probability distributions,\footnote{We will refer to probability distributions simply as distributions in the remainder of the paper.} while lowercase sans-serif letters  (e.g., \outpdf) are reserved for probability density functions. 
The superscripts~$\tp{}$ and~$\herm{}$ stand for
transposition and Hermitian transposition,
respectively. 
We denote the identity matrix of dimension~$\blocklength\times \blocklength$ by~$\matI_{\blocklength}$; $\diag\{\veca\}$ is the diagonal square matrix whose main diagonal contains the entries of the vector~\veca, and $\eig_{q}(\matA)$ stands for the $q$th largest eigenvalue of the Hermitian positive-semidefinite matrix~\matA.
For a random vector~$\vecx$ with distribution~$\inpprobmeas$, we
write~$\vecx\distas\inpprobmeas$. 
We denote expectation by~$\Ex{}{\cdot}$, and use the
notation~$\Ex{\vecx}{\cdot}$ or $\Ex{\inpprobmeas}{\cdot}$ to stress that expectation is taken with
respect to $\vecx\distas \inpprobmeas$. 
We write~$\relent{\inpprobmeas(\cdot)}{\outprobmeas(\cdot)}$ for the
relative entropy between the distributions~$\inpprobmeas$
and~$\outprobmeas$~\cite[Sec.~8.5]{cover06-a}. 
Furthermore, $\jpg(\veczero,\corrmat)$ stands for the distribution of a
circularly-symmetric~\cite[Def.~24.3.2]{lapidoth09a} complex Gaussian random vector with covariance
matrix~$\corrmat$.
For two functions~$f(x)$ and~$g(x)$, the
notation~$f(x) = \landauO(g(x))$, $x\to \infty$, means that
$\lim \sup_{x\to\infty}\bigl|f(x)/g(x)\bigr|<\infty$, and
$f(x) = \landauo(g(x))$, $x\to \infty$, means that $\lim_{x\to\infty}\bigl|f(x)/g(x)\bigr|=0$.    
Finally, $\log(\cdot)$ indicates the natural logarithm.
 

\subsection{The Channel Model} 
\label{sec:the_channel_model}
%
In our quest for simplicity of exposition, we chose to focus on the correlated block-fading channel model~\cite{liang04-09a,morgenshtern10-06a}.
In this model, the channel changes in an independent fashion across blocks of \blocklength discrete-time samples and exhibits correlated fading within each block (with the same fading statistics for all blocks).
The input-output (IO) relation corresponding to one such block is given by:%
%
\begin{equation}
	\label{eq:io_general}
	\outvec=\diag\{\chgenvec\} \inpvec +\wgnvec.
\end{equation}
Here, $\inpvec=\tp{[\inp_1 \dots \inp_{\blocklength}]} \in \complexset^{\blocklength}$ is the (random) input vector, which we assume to satisfy the average-power constraint 
\begin{equation}\label{eq:avp}
	\frac{1}{\blocklength}\Ex{}{\vecnorm{\inpvec}^2}\leq  \snr.
\end{equation}
The vector~$\wgnvec\distas\jpg(\veczero,\matI_{\blocklength})$ represents additive white Gaussian noise (AWGN), and~$\chgenvec\distas\jpg(\veczero,\corrmat)$ contains the fading channel coefficients.
The vectors \inpvec, \chgenvec, and \wgnvec are mutually independent. 
We assume that~\corrmat has rank~\rankcorr ($1\leq \rankcorr\leq \blocklength$) and that  the main-diagonal entries of~\corrmat are all equal to $1$. 
%
Throughout the paper, we consider the noncoherent setting where transmitter and receiver know the statistics of 
\chgenvec, but not its realizations.

The model we just described may seem contrived at first sight.
Yet, it is of practical relevance for at least two reasons.
First, it captures the essence of  channel variations (in time) in an accurate but simple way: 
the rank 
\rankcorr of \corrmat corresponds to the minimum number of entries of \chgenvec that need to be known to perfectly recover the whole vector (in the absence of noise);
therefore, larger \rankcorr corresponds to faster channel variation. 
Second, when \corrmat is circulant, the IO relation~\eqref{eq:io_general} coincides with the IO relation---in the frequency domain---of a cyclic-prefix orthogonal frequency-division multiplexing system~\cite{peled80-04a} operating over a frequency-selective channel with $\rankcorr$ uncorrelated taps.
In other words, the model in~\eqref{eq:io_general} can be thought of as the dual of the widely used intersymbol-interference channel model.
Independence across blocks is a sensible assumption for systems employing time-division multiple access or frequency hopping~\cite{marzetta99-01a}.
Finally, we remark that for the special case $\rankcorr=1$, the channel model in~\eqref{eq:io_general} reduces to the \emph{piecewise-constant} block-fading channel model previously used in numerous papers such as~\cite{marzetta99-01a,hochwald00-03a,zheng02-02a}.
%

\subsection{Channel Capacity} 
\label{sec:channel_capacity}
The capacity of the channel in~\eqref{eq:io_general} is given by
\begin{equation}\label{eq:capacity}
	\capacity(\snr)=\frac{1}{\blocklength}\sup_{\inpprobmeas} \mi(\inpvec;\outvec).
\end{equation}
Here, $\mi(\inpvec;\outvec)$ denotes the mutual information~\cite[Sec. 8.5]{cover06-a} between \inpvec and \outvec in~\eqref{eq:io_general}, and the supremum is taken over all distributions~\inpprobmeas on~\inpvec that satisfy the average-power constraint~\eqref{eq:avp}.
Because the variance of the entries of~\chgenvec and \wgnvec is normalized to one, we can interpret \snr as the receive SNR.

The literature is essentially void of analytic expressions for $\capacity(\snr)$, even for the simplest case $
\blocklength=1$.
Nevertheless, as we shall see in the next section, the high-SNR behavior of $\capacity(\snr)$ can be characterized fairly well.

 \subsection{Known Results and Our Contributions} 
 \label{sec:known_results_and_outline}
For the \emph{general case} $1\leq \rankcorr \leq \blocklength$, Liang and Veeravalli showed that
\cite[Props.~3 and~4]{liang04-09a}
\begin{equation}\label{eq:liang_capacity_result}	
 \capacity(\snr)= \frac{\blocklength-\rankcorr}{\blocklength}\log\snr +\landauO(\log\log\snr), \quad \snr \to \infty.
\end{equation}
This result is sufficient to characterize the capacity pre-log \prelog, defined as the asymptotic ratio between  capacity and the logarithm of SNR as SNR goes to infinity:
 \begin{equation*}
 	\prelog=\lim_{\snr \to \infty} \dfrac{\capacity(\snr)}{\log \snr}.
 \end{equation*}
The pre-log can be interpreted as the fraction of signal-space dimensions that can be used for communication.
From~\eqref{eq:liang_capacity_result} we find the pre-log to be given by the difference of two terms, i.e.,  $\prelog=1-\rankcorr/\blocklength$.
The first term can be thought of as the capacity pre-log when the channel is known perfectly at the receiver (in this case, $\prelog=1$~\cite{biglieri98-10a}); the second term quantifies the loss in signal-space dimensions due to the lack of channel knowledge. 
Note that $\rankcorr/\blocklength$ is the smallest fraction of entries of the $\blocklength$-dimensional vector \chgenvec that need to be known to reconstruct the whole vector in the absence of noise.\footnote{As we shall see, neglecting additive noise in~\eqref{eq:io_general} yields useful insights on the capacity pre-log.}       
Hence, we can further interpret the penalty term $\rankcorr/\blocklength$ as the fraction of signal-space dimensions in which \emph{pilot symbols} need to be transmitted to allow the receiver to learn the channel.

When $\rankcorr=\blocklength$, i.e., the channel correlation matrix has full rank, \eqref{eq:liang_capacity_result} implies that the pre-log is equal to $0$.
It turns out that in this case the $\landauO(\log\log\snr)$ term in~\eqref{eq:liang_capacity_result}  is tight and capacity grows double-logarithmically in SNR.
This  surprising result was proven in~\cite[Lem.~5]{liang04-09a}.
In~\fref{sec:the_full_rank_case}, we shall refine the result in~\cite[Lem.~5]{liang04-09a} by providing the following, more accurate, high-SNR capacity characterization:
\begin{multline}\label{eq:capacity_high_snr_full_rank}
 	\capacity(\snr)=\log\log \snr -\gamma -1 \\ -\frac{1}{\blocklength}\sum_{q=1}^\blocklength\log\eig_q(\corrmat)+\landauo(1), \quad \snr\to\infty.
 \end{multline}
This result characterizes capacity (for $\rankcorr=\blocklength$) up to a $\landauo(1)$ term (i.e., a term that vanishes as $
\snr \to \infty$).
In contrast, the expression provided in~\cite[Lem.~5]{liang04-09a} agrees with capacity only up to a $\landauO(1)$ term (i.e., a term that is bounded as $
\snr \to \infty$).    

The most important tool in the proof of~\eqref{eq:capacity_high_snr_full_rank} is the \emph{duality approach}, a technique first introduced in~\cite{lapidoth03-10a} to characterize the capacity of stationary ergodic fading channels  with finite differential entropy rate.
The essence of the duality approach is that it allows one to obtain tight upper bounds on $\capacity(\snr)$ by choosing \emph{appropriate distributions} on the output~\outvec.
Compared to the treatment in~\cite{lapidoth03-10a}, our goal in Sections \ref{sec:duality} and \ref{sec:the_full_rank_case} is to provide the simplest and most accessible proofs for the main results underlying the duality approach.
This comes at the cost of generality (in terms of noise and fading statistics). 

While finding a capacity characterization that---like~\eqref{eq:capacity_high_snr_full_rank}---is tight up to a $\landauo(1)$ term for all \rankcorr with $1\leq \rankcorr \leq \blocklength$ is an interesting open problem,
for the special case $\rankcorr=1$ (with $\rankcorr<\blocklength$) the following result was reported in~\cite{hochwald00-03a,zheng02-02a}: 
\begin{multline}\label{eq:known_result_rank_1}
\capacity(\snr)=\frac{\blocklength-1}{\blocklength}\bigl[ \log\snr +\log\blocklength-\euler-1 \bigr] \\- \frac{\log \Gamma(\blocklength)}{\blocklength} +\landauo(1),\quad \snr \to \infty.
\end{multline}
Here, $\gamma$ denotes the Euler-Mascheroni constant, and $\Gamma(\cdot)$ is the Gamma function~\cite[Eq.~(197)]{lapidoth03-10a}.
The proof of~\eqref{eq:known_result_rank_1} provided in~\cite{hochwald00-03a} is based on a rather technical argument and does not seem to explicitly exploit the  geometry in the problem, i.e., the fact that \inpvec and \outvec are collinear in the absence of noise.
The proof in~\cite{zheng02-02a} does exploit this geometry through an apposite change of variables, and applies to the multiple-antenna setting as well.

In \fref{sec:the_rank_1_case}, we present a simple, alternative proof of~\eqref{eq:known_result_rank_1} that, differently from the proofs in~\cite{hochwald00-03a,zheng02-02a}, is based on duality and exploits the geometry in the problem to motivate the choice of the output distribution.
Our proof needs another tool put forward in~\cite{lapidoth03-10a}: the \emph{escape to infinity} property of the capacity achieving distribution.
This property, which we review in~\fref{sec:escape_to_infinity}, allows one to restrict the maximization in~\eqref{eq:capacity} to a smaller set of distributions.

\section{The Toolbox} 
\label{sec:the_toolbox}%
\subsection{The Duality Approach} 
\label{sec:duality}%
To prove~\eqref{eq:capacity_high_snr_full_rank}  and~\eqref{eq:known_result_rank_1}, we sandwich capacity between a lower and an upper bound that agree up to a $\landauo(1)$ term.
Establishing capacity lower bounds is, in principle, relatively simple: it suffices to evaluate the mutual information in~\eqref{eq:capacity} for an input distribution~\inpprobmeas that satisfies the average-power constraint.
Obviously, care must be exercised in choosing~\inpprobmeas, so as to ensure that the resulting bound is tight in the limit $\snr \to \infty$ (see~\fref{sec:a_capacity_lower_bound_rank_1} for a concrete example).

Capacity upper bounds are more difficult to find because of the need for
maximization over the set of eligible input distributions. 
To single out the main difficulty with this optimization problem, it is convenient to denote the conditional distribution of \outvec given \inpvec  as  $\chtran(\cdot\given\inpvec)$ and to use the symbol $\inpprobmeas\chtran$ to indicate the distribution induced on \outvec by the input distribution $\inpprobmeas$ and by the ``channel'' $\chtran(\cdot\given\inpvec)$.
By the definition of mutual information~\cite[Sec. 8.5]{cover06-a} we have that
\begin{align}\label{eq:mutual_information_rel_entropy}
	\mi(\inpvec;\outvec)=\Ex{\inpprobmeas}{\relent{\chtran(\cdot\given\inpvec)}{(\inpprobmeas\chtran)(\cdot)}}.
\end{align}
As the right-hand side (RHS) of~\eqref{eq:mutual_information_rel_entropy}  is a rather complicated function of~\inpprobmeas,  the maximization in~\eqref{eq:capacity} is difficult to carry out.
The idea behind duality is to upper-bound the RHS of~\eqref{eq:mutual_information_rel_entropy} by replacing $\inpprobmeas\chtran$ by a distribution that does not depend on $\inpprobmeas$. 
Concretely, let $\outprobmeas$ be an arbitrary distribution on \outvec.    
Then
\begin{align}\label{eq:duality_ub}
		\mi(\inpvec;\outvec)&\stackrel{(a)}{=}
		%
		%
		\Ex{\inpprobmeas}{\relent{\chtran(\cdot\given \inpvec)}{\outprobmeas(\cdot)}} - \relent{(\inpprobmeas\chtran)(\cdot)}{\outprobmeas(\cdot)} \notag\\
		&\stackrel{(b)}{\leq} \Ex{\inpprobmeas}{\relent{\chtran(\cdot\given \inpvec)}{\outprobmeas(\cdot)}}.
 \end{align}
Here, (a) follows from Tops{\o}e's identity~\cite{topsoe67-a} and (b) is a consequence of the nonnegativity of relative entropy~\cite[Thm.~2.6.3]{cover06-a}.     
The RHS of~\eqref{eq:duality_ub}  is easier to deal with than $ \mi(\inpvec;\outvec)$.
In fact, as we shall illustrate in  Sections~\ref{sec:the_rank_1_case} and \ref{sec:the_full_rank_case}, it is possible---for an appropriate choice of  $\outprobmeas$---to find an asymptotically tight upper bound on $\Ex{\inpprobmeas}{\relent{\chtran(\cdot\given \inpvec)}{\outprobmeas(\cdot)}}$ that holds for every \inpprobmeas satisfying the average-power constraint~\eqref{eq:avp}.
By~\eqref{eq:capacity}, this upper bound constitutes an upper bound on $\capacity(\snr)$.

As a side remark, we note that the inequality~\eqref{eq:duality_ub} holds with equality when~$\outprobmeas$ coincides with $\inpprobmeas\chtran$.
Hence,~\eqref{eq:duality_ub} yields the following  expression for mutual information:
\begin{equation}\label{eq:duality_for_mutual_information}	
	\mi(\inpvec;\outvec)= \inf_{\outprobmeas}\Ex{\inpprobmeas}{
	\relent{\chtran(\cdot\given\inpvec)}{\outprobmeas(\cdot)}
	}.
\end{equation}
Through further manipulations (see~\cite{csiszar82-a,lapidoth03-10a} for details), the identity~\eqref{eq:duality_for_mutual_information} yields a \emph{dual} expression for capacity, with the maximization over the input distribution in~\eqref{eq:capacity}  replaced by a minimization over the output distribution.
This is why the technique is referred to as duality approach.    

An appropriate choice of the output distribution~$\outprobmeas$ is crucial for the bound in~\eqref{eq:duality_ub}  to be tight.
Throughout the paper, the output distribution \outprobmeas with density
\begin{equation}\label{eq:generic_output_distribution}
	\outpdf(\outvec)= \frac{\Gamma(\blocklength)}{\pi^\blocklength \beta^{\alpha}\Gamma(\alpha)} \vecnorm{\outvec}^{2(\alpha-\blocklength)} e^{-\vecnorm{\outvec}^2/\beta}, \quad \outvec \in \complexset^{\blocklength}
\end{equation}
will play a prominent role. 
Here, $\beta=\blocklength(
\snr+1)/\alpha$, where $\alpha$ is a free parameter whose meaning will become clear later.  
This output distribution was put forward in~\cite{lapidoth03-10a} in a more general setting.
The main features of this distribution are that $\outvec$ is isotropically distributed and that $\vecnorm{\outvec}^2$  is Gamma distributed with parameter $\alpha$. 
In \fref{sec:the_rank_1_case}, we will show that, for the piecewise-constant block-fading channel model (i.e., $\rankcorr=1$), this choice for the output distribution can be motivated through simple geometric intuition.                   
 \subsection{Escape-To-Infinity Property} 
 \label{sec:escape_to_infinity}    
Duality simplifies the maximization over the input distribution in~\eqref{eq:capacity}, at the cost of getting an  upper bound on capacity.
This simplification, together with an appropriate choice of \inpprobmeas to obtain a matching capacity lower bound, is enough to establish~\eqref{eq:capacity_high_snr_full_rank}, as we shall see in \fref{sec:the_full_rank_case}. 
To prove~\eqref{eq:known_result_rank_1}, however, we need an additional tool.
Specifically, we will make use of the fact that the asymptotic behavior of $\capacity(\snr)$ does not change if we constrain the input distributions \inpprobmeas in~\eqref{eq:capacity} to be supported strictly outside a sphere of arbitrarily large radius.  
We formalize this result, which turns out to hold for almost all wireless channel models of practical interest~\cite{lapidoth03-10a,lapidoth06-02a}, in the following theorem.
In view of~\eqref{eq:known_result_rank_1},  we focus on the case $\rankcorr=1$.
\begin{thm}\label{thm:outside_sphere}
	   Fix an arbitrary $\snr_0>0$ and let $\setK=\{\inpvec \in \complexset^{\blocklength} \sothat \vecnorm{\inpvec}^2\leq \snr_0\}$.   
	Denote by $\capacity(\snr)$ the capacity of the channel with IO relation~\eqref{eq:io_general} (with $\rankcorr=1$)
	under the average-power constraint~\eqref{eq:avp}.
	Furthermore, denote by $\capacityconstr(\snr)$ the capacity of the same channel under the additional constraint---besides~\eqref{eq:avp}---that $\inpvec \notin \setK$ with probability one (\wpone).
	Then 
	\begin{multline*}	
		\lim_{\snr \to \infty} \Bigl[ \capacity(\snr) - \bigl(1-{1}/{\blocklength}\bigr)\log \snr \Bigr]   \\
		=      \lim_{\snr \to \infty} \Bigl[ \capacityconstr(\snr) - \bigl(1-{1}/{\blocklength}\bigr)\log \snr \Bigr].
	\end{multline*}
\end{thm} 
\begin{proof}
The high-SNR capacity expansion~\eqref{eq:liang_capacity_result} implies that the capacity pre-log is\footnote{It is worth mentioning that the proof of~\eqref{eq:liang_capacity_result} does not make use of~\fref{thm:outside_sphere}; so there is no cyclic argument here.} $1-1/\blocklength$.
The logarithmic growth of capacity in SNR allows us to invoke~\cite[Thm.~8]{lapidoth06-02a} and conclude that the capacity-achieving input distribution must \emph{escape to infinity}~\cite[Def.~4.11]{lapidoth03-10a}, i.e., that for all  $\snr_0 \geq 0$ there exists a family of input distributions  $\{\inpprobmeas_{\snr}\}_{\snr\geq 0}$ (parametrized with respect to \snr)  satisfying $(1/\blocklength)\Ex{\inpprobmeas_\snr}{\vecnorm{\inpvec}^2}\leq \snr$, such that, when $\inpvec \distas \inpprobmeas_\snr$,
	\begin{equation*}
		\lim_{\snr \to \infty} \{ \capacity(\snr) - \mi(\inpvec;\outvec) \}=0
	\end{equation*}
	and
	\begin{equation*}
		\lim_{\snr \to \infty}  \Prob\{\vecnorm{\inpvec}^2\leq \snr_0 \}=0.
	\end{equation*}                                       
	The proof is concluded by noting that the escape-to-infinity property is a sufficient condition for \fref{thm:outside_sphere} to hold, as a consequence of \cite[Thm.~4.12]{lapidoth03-10a}.
\end{proof}
\section{High-SNR Capacity Asymptotics} 
\label{sec:high_snr_results}
\subsection{The Rank-One Case} 
\label{sec:the_rank_1_case}     
When $\rankcorr=1$, we can rewrite~\eqref{eq:io_general} in the following (more convenient) form
\begin{equation}
	\label{eq:io_rank1}
	\outvec=\ch\inpvec +\wgnvec
\end{equation}
where~$\ch\distas\jpg(0,1)$.  
The high-SNR capacity expansion~\eqref{eq:liang_capacity_result} implies that the capacity pre-log of the channel in~\eqref{eq:io_rank1} is  given by $1-1/\blocklength$.
This is in agreement with the intuition we provided in~\fref{sec:known_results_and_outline}: one pilot symbol per block is enough to learn the channel in the absence of noise.
We next provide a different interpretation of this result, which is of geometric nature and sheds light on how to select input and output distributions to get capacity bounds that are tight as $\snr \to \infty$.

 \subsubsection{Geometric Intuition} 
 \label{sec:geometric_intuition}
Let \inpvec be an arbitrary vector in $\complexset^{\blocklength}$.    
This vector can be specified by identifying
 \begin{inparaenum}[i)]
 	\item the linear subspace spanned by \inpvec, i.e., the complex line passing through the origin and \inpvec   and
 	\item the point on that line corresponding to \inpvec (i.e., a complex number).
 \end{inparaenum}
If we neglect additive noise, the IO relation in~\eqref{eq:io_rank1}  reduces to $\outvec=\ch\inpvec$.  
As~\ch varies,~\outvec spans the line~\inpvec lies on.
 In other words---as pointed out in~\cite{zheng02-02a}---the random channel coefficient~\ch
  destroys the information about~\inpvec specified in the second step of our description above,
but leaves the information about the linear subspace spanned by \inpvec unchanged.
To summarize, when the random channel coefficient~\ch is not known to the receiver, the information that the receiver can recover about the transmitted signal~\inpvec is the line on which~\inpvec lies. 
 But a complex line in~$\complexset^{\blocklength}$ is fully characterized by $\blocklength-1$ complex parameters.\footnote{More formally, the set of lines passing through the origin of~$\complexset^{\blocklength}$ forms a manifold (the complex projective space $\complexset\setP^{N-1}$) of $N-1$ complex dimensions~\cite{boothby86-a}.} 
Hence, the received signal ``carries'' $\blocklength-1$ parameters describing \inpvec.
This number, divided by \blocklength, coincides with the capacity pre-log.
%
%
%


\subsubsection{A Capacity Lower Bound} 
\label{sec:a_capacity_lower_bound_rank_1}
The geometry unveiled in the previous section suggests to use the direction of \inpvec, but not its magnitude, to convey information.
This insight is helpful in choosing an input distribution that yields a tight capacity lower bound.
Concretely, we take  $\inpvec=\sqrt{\blocklength\snr}\cdot\id{\inpvec}$ where $\id{\inpvec}$ is uniformly distributed on the unit sphere in $\complexset^{\blocklength}$. 
We use this input distribution, which trivially satisfies the average-power constraint~\eqref{eq:avp}, to lower-bound capacity as follows:                                          
\begin{align}
	\blocklength \cdot \capacity(\snr)&\geq \mi(\inpvec;\outvec)= \difent(\outvec)-\difent(\outvec\given \inpvec) \notag\\
	&\geq \difent(\outvec \given \wgnvec) -   \difent(\outvec\given \inpvec) \notag\\  
	&= \difent(\underbrace{\ch\inpvec}_{\define\outvecnn})-\difent(\outvec\given\inpvec).\label{eq:diff_ent_difference}
\end{align}
Here, $\difent(\cdot)$ denotes differential entropy~\cite[Sec.~8.1]{cover06-a}, the second inequality follows because conditioning reduces differential entropy~\cite[Sec.~8.6]{cover06-a}, and the last equality follows because differential entropy is invariant to translations~\cite[Thm.~8.6.3]{cover06-a} and \wgnvec is independent of \ch and \inpvec.
To compute $\difent(\outvecnn)$, it is convenient to switch to polar coordinates, i.e., $\outvecnn \mapsto (\vecnorm{\outvecnn},\id{\outvecnn})$, where $\id{\outvecnn}=\outvecnn/\vecnorm{\outvecnn}$.
The change of variable theorem then yields~\cite[Lem.~6.17]{lapidoth03-10a}:
\begin{equation}
	\label{eq:diff_ent_and_change_of_variable}
	\difent(\outvecnn)=\difent(\vecnorm{\outvecnn})+\difentsphere(\id{\outvecnn}\given \vecnorm{\outvecnn})+(2\blocklength-1)\Ex{}{\log\vecnorm{\outvecnn}}.
\end{equation}
Here,~$\difentsphere(\cdot)$ denotes the differential entropy computed with respect to the area measure on the unit sphere in~$\complexset^\blocklength$ \cite[p.~2457]{lapidoth03-10a}.
%
By the choice of the input distribution, we have $\vecnorm{\outvecnn}=\sqrt{\blocklength\snr}\abs{\ch}$.
Furthermore, because~$\id{\inpvec}$ is uniformly distributed on the unit sphere and~$\ch$ is circularly symmetric (i.e., the phase of \ch is uniformly distributed on $[-\pi,\pi)$ and is independent of $\abs{\ch}$~\cite[Prop. 24.2.6]{lapidoth09a}), it follows that \outvecnn is isotropically distributed~\cite[Def.~6.19]{lapidoth03-10a}.
Hence, $\id{\outvecnn}$ is uniformly distributed on the unit sphere and is independent of $\vecnorm{\outvecnn}$.
Based on these observations, we can now simplify~\eqref{eq:diff_ent_and_change_of_variable} as follows:
 \begin{align}
 	\label{eq:diff_ent_simplified}
 	\difent(\outvecnn)
\stackrel{(a)}{=}&\difent\bigl(\sqrt{\blocklength\snr}\abs{\ch}\bigr)+ \difentsphere(\id{\outvecnn})+(2\blocklength-1)\Ex{}{\log\vecnorm{\outvecnn}} \notag\\
\stackrel{(b)}{=}&\log\sqrt{\blocklength \snr}+\difent(\abs{\ch})+\log\dfrac{2\pi^{\blocklength}}{\Gamma(\blocklength)} \notag
\\&+
(2\blocklength-1)\bigl[\log\sqrt{\blocklength\snr}+ \Ex{}{\log\abs{\ch}}\bigr] \notag \\
\stackrel{(c)}{=}& \blocklength\log(\blocklength \snr) +\difent\bigl(\abs{\ch}^2\bigr)+\log \dfrac{\pi^{\blocklength}}{\Gamma(\blocklength)} \notag
\\&
+(\blocklength-1)\Ex{}{\log\abs{\ch}^2}\notag\\
\stackrel{(d)}{=}&\blocklength\log(\blocklength \snr) +1 + \log \dfrac{\pi^{\blocklength}}{\Gamma(\blocklength)} -(\blocklength-1)\euler.
 \end{align}
Here, in~(a) we used the independence of  $\id{\outvecnn}$ and $\vecnorm{\outvecnn}$ to drop  conditioning in the second term on the RHS of~\eqref{eq:diff_ent_and_change_of_variable}.
In~(b) we used that $\difent(ax)=\log a+\difent(x)$ for $x$ a real-valued random variable and $a$ real and nonnegative;
we also used that  $\id{\outvecnn}$  is uniformly distributed on the unit sphere in~$\complexset^{\blocklength}$, and that, as a consequence, $\difentsphere(\id{\outvecnn})$  is equal to the area of that sphere, i.e., $2\pi^{\blocklength}/\Gamma(\blocklength)$. 
In~(c) we used that 
\begin{equation*}	
	\difent(v)=\difent(v^2)-\Ex{}{\log v} -\log 2
\end{equation*}
for every real nonnegative random variable~$v$ \cite[Lem.~6.15]{lapidoth03-10a},
and~(d) follows because 
 $\Ex{}{\log{\abs{\ch}^2}}=-\euler$ and $\difent(\abs{\ch}^2)=1$ for $\ch\distas \jpg(0,1)$.

Since~\chgenvec is a circularly-symmetric complex Gaussian vector, $\difent(\outvec\given\inpvec)$ on the RHS of~\eqref{eq:diff_ent_difference} admits the following closed-form expression:
\begin{align}
	\label{eq:conditional_diff_entropy}
	\difent(\outvec\given \inpvec)&= \log(\pi e)^\blocklength +\Ex{}{\log(1+\vecnorm{\inpvec}^2)} \notag\\
	&=\log(\pi e)^\blocklength + \log(1+\blocklength \snr) \notag \\
	&= \log(\pi e)^\blocklength + \log (\blocklength \snr) + \landauo(1), \quad \snr \to \infty.
\end{align}
Substituting~\eqref{eq:diff_ent_simplified} and~\eqref{eq:conditional_diff_entropy} into~\eqref{eq:diff_ent_difference}, we get a lower bound on $\capacity(\snr)$ that coincides with the RHS of~\eqref{eq:known_result_rank_1}.

\subsubsection{A Matching Upper Bound} 
\label{sec:a_matching_upper_bound_rank_1}
To obtain an upper bound that matches the lower bound we just found up to a $\landauo(1)$ term, we use duality and the escape-to-infinity property of the capacity-achieving distribution.  
More specifically, as a consequence of~\fref{thm:outside_sphere}, we can, without loss of generality, constrain the maximization of  mutual information in~\eqref{eq:capacity} to  input distributions~\inpprobmeas that satisfy---besides the average-power constraint~\eqref{eq:avp}---the additional constraint $\vecnorm{\inpvec}^2\geq \snr_0$ \wpone.
Here, $\snr_0>0$ is a parameter to be optimized later.
We use duality with the density of \outprobmeas given by~\eqref{eq:generic_output_distribution} with $\alpha=1$.
This choice is again motivated by the geometric considerations in \fref{sec:geometric_intuition}: in the noiseless case, the  density of the output distribution induced by the input distribution used in \fref{sec:a_capacity_lower_bound_rank_1} (to derive a capacity lower bound), equals~\eqref{eq:generic_output_distribution} with $\alpha=1$.
Fix $\snr_0>0$ and take an arbitrary input distribution \inpprobmeas such that $\inpvec\distas\inpprobmeas$ satisfies~\eqref{eq:avp} and  $\vecnorm{\inpvec}^2\geq \snr_0$ \wpone.
By duality, we have that
\begin{align}
	\label{eq:upper_bound_duality_1}
	\mi(\inpvec;\outvec) \leq& \Ex{\inpprobmeas}{\relent{\chtran(\cdot\given \inpvec)}{\outprobmeas(\cdot)}} \notag\\
	=&\Ex{\inpprobmeas\chtran}{\log \frac{1}{\outpdf(\outvec)}}-\difent(\outvec \given \inpvec)    \notag \\     
	=&\log \pi^N + \log[\blocklength(\snr+1)]-\log \Gamma(\blocklength)\notag \\  
	&+ (\blocklength-1) \Ex{\inpprobmeas\chtran}{\log \vecnorm{\outvec}^2}  
	+ \frac{\Ex{\inpprobmeas\chtran}{\vecnorm{\outvec}^2}}{\blocklength(\snr+1)}
	- \difent(\outvec \given \inpvec). 
\end{align}
Here, the first equality follows from straightforward algebraic manipulations; in the second equality we used~\eqref{eq:generic_output_distribution} with $\alpha=1$.
We shall next evaluate or bound the terms on the RHS of~\eqref{eq:upper_bound_duality_1} that depend on~\inpprobmeas.
First, note that
\begin{equation}
	\Ex{\inpprobmeas\chtran}{\vecnorm{\outvec}^2} = \Ex{\inpprobmeas\chtran}{\vecnorm{\ch\inpvec+\wgnvec}^2} \leq \blocklength(\snr+1).
\end{equation}
Here, we used independence of \inpvec and \wgnvec and the power constraint~\eqref{eq:avp}.
To evaluate $\Ex{\inpprobmeas\chtran}{\log \vecnorm{\outvec}^2}$, we proceed as follows: 
first note that
\begin{equation*}	
	   \Ex{\inpprobmeas\chtran}{\log \vecnorm{\outvec}^2} =   \Ex{\inpvec}{\Ex{\ch,\wgnvec}{\log \vecnorm{\outvec}^2 \big |\, \inpvec}}.
\end{equation*}
%
We next use that, given~\inpvec, the random variable $\vecnorm{\outvec}^2$ is distributed as $\sum_{i=1}^{\blocklength-1}\abs{\altwgn_i}^2 +(1+\vecnorm{\inpvec}^2)\abs{\altwgn_{\blocklength}}^2$, where the $\altwgn_i$, $\allo{i}{\blocklength}$, are \iid $\jpg(0,1)$. 
This result follows by observing that, given \inpvec, the output vector \outvec has covariance matrix $\inpvec\herm{\inpvec}+\matI_{\blocklength}$ (whose eigenvalues are $1+\vecnorm{\inpvec}^2$ and $1$ with multiplicity $\blocklength-1$).
Using Jensen's inequality with respect to the random variables~$\altwgn_1, \dots, \altwgn_{\blocklength-1}$, we  obtain the following bound:
\begin{align}
	&\Ex{\inpvec}{\Ex{\ch,\wgnvec}{\log \vecnorm{\outvec}^2\big\lvert\, \inpvec}} \notag\\
	&= \Ex{\inpvec}{\Ex{\altwgn_1, \dots, \altwgn_\blocklength}{\log\Biggl( \sum_{i=1}^{\blocklength-1}\abs{\altwgn_i}^2 +(1+\vecnorm{\inpvec}^2)\abs{\altwgn_{\blocklength}}^2\Biggr)\biggl\vert\, \inpvec} } \notag \\
	&\leq \Ex{\inpvec,\altwgn_{\blocklength}}{\log\bigl(\blocklength-1 + (1+\vecnorm{\inpvec}^2)\abs{\altwgn_{\blocklength}}^2\bigr)}  \notag\\
	&=  \Ex{\inpvec}{\log(1+\vecnorm{\inpvec}^2)} +\Ex{\inpvec,\altwgn_{\blocklength}}{ \log\lefto(\frac{\blocklength-1}{1+\vecnorm{\inpvec}^2}+\abs{\altwgn_{\blocklength}}^2\right)}\label{eq:before_lemma_intermediate}\\
	&\leq\Ex{\inpvec}{\log(1+\vecnorm{\inpvec}^2)} \notag\\
	&\quad+\sup_{\vecnorm{\inpvec}^2\geq \snr_0}\Ex{\altwgn_{\blocklength}}{ \log\lefto(\frac{\blocklength-1}{1+\vecnorm{\inpvec}^2}+\abs{\altwgn_{\blocklength}}^2\right)}.
	\label{eq:before_lemma}
\end{align} 
In the last step, we upper-bounded the second term on the RHS of~\eqref{eq:before_lemma_intermediate} by replacing the expectation over $\inpvec$ by the supremum over all vectors \inpvec satisfying $\vecnorm{\inpvec}^2\geq\snr_0$.      
This is the step where the escape-to-infinity property is used.  
Without this property, the supremum in~\eqref{eq:before_lemma} would be over all \inpvec satisfying $\vecnorm{\inpvec}^2\geq 0$ and the resulting bound would not match (up to a $\landauo(1)$ term) the lower bound obtained in \fref{sec:a_capacity_lower_bound_rank_1}. 
To evaluate the second term on the RHS of~\eqref{eq:before_lemma} we use the following lemma:
\begin{lem}\label{lem:exp_log}
	Let $z\distas \jpg(0,1)$ and take $a>0$. 
	Then
	\begin{align*}
		\Ex{z}{\log(a+\abs{z}^2)} =\underbrace{e^a \Gamma(0,a)+\log a}_{\define g(a)}.
	\end{align*}
	Here $\Gamma(\cdot,\cdot)$ denotes the incomplete Gamma function~\cite[Eq.~(200)]{lapidoth03-10a}.
	The function $g(a)$ is monotonically increasing in $a$.
	Furthermore, $\lim_{a\to 0}g(a)=-\gamma$.
\end{lem}
\begin{proof}
	    Let $v=\abs{z}^2$.
	Then
	\begin{align}
		      \Ex{z}{\log(a+\abs{z}^2)} &= \int_{0}^{\infty} e^{-v} \log(a+v) dv \nonumber\\
		&= e^a \int_a^{\infty} e^{-t} \log t\, dt   \label{eq:lemma_insert_a} \\ 
		&= e^{a}\Gamma(0,a)+\log a. \label{eq:lemma_final}
	\end{align}
	Here, to obtain the second equality we used integration by parts.
	Now let us denote the RHS of~\eqref{eq:lemma_final} by $g(a)$.
	It is easy to verify that $g(a)$ is a monotonic function of $a\geq 0$.
	In fact,
	\begin{equation*}	
		\frac{dg(a)}{da}=e^a \Gamma(0,a)
	\end{equation*}
	which is nonnegative for all $a\geq 0$.
	Finally, the claim that $\lim_{a\to 0}g(a)=-\gamma$ follows from~\eqref{eq:lemma_insert_a} by setting $a=0$. 
\end{proof}
As a consequence of~\fref{lem:exp_log}, we have that
\begin{align*}
	\sup_{\vecnorm{\inpvec}^2\geq \snr_0}\Ex{\altwgn_{\blocklength}}{ \log\lefto(\frac{\blocklength-1}{1+\vecnorm{\inpvec}^2}+\abs{\altwgn_{\blocklength}}^2\right)}=g\lefto(\frac{\blocklength-1}{1+\snr_0} \right).
\end{align*}
Finally,  for the conditional differential entropy term in~\eqref{eq:upper_bound_duality_1} we have 
\begin{align*}
	\difent(\outvec \given \inpvec) =& \log(\pi e)^\blocklength + \Ex{}{\log(1+\vecnorm{\inpvec}^2)}.
\end{align*}
To summarize, we proved that
\begin{align*}
	\mi(\inpvec;\outvec) \leq & \log(\blocklength\snr+\blocklength) +
	  (\blocklength-2)\Ex{}{\log(1+\vecnorm{\inpvec}^2)}\\
	  &+(\blocklength-1)g\lefto(\frac{\blocklength-1}{1+\snr_0} \right) 
	-\log\Gamma(\blocklength) -(\blocklength-1).
\end{align*}
%
Now, using Jensen's inequality on $\Ex{}{\log(1+\vecnorm{\inpvec}^2)}$, we obtain                               
\begin{multline*}	
	\lim_{\snr \to \infty} \bigl[\mi(\inpvec;\outvec)-(\blocklength-1)\log \snr \bigr] \\
	\leq (\blocklength-1) 
	\left[ \log \blocklength -1 + g\lefto(\frac{\blocklength-1}{1+\snr_0}\right)\right] -\log\Gamma(\blocklength).
\end{multline*}
The proof is concluded by recalling that, by \fref{thm:outside_sphere}, the asymptotic behavior of $\capacity(\snr)$ does not change if we constrain \inpprobmeas to satisfy $\vecnorm{\inpvec}^2\geq \snr_0$ \wpone,  and by noting that, by \fref{lem:exp_log}, we can make the term $g((\blocklength-1)/(1+\snr_0))$  to be arbitrarily close to $-\gamma$ by taking $\snr_0$ sufficiently large.

\subsection{The Full-Rank Case} 
\label{sec:the_full_rank_case}
Due to space constraints, we shall  give an outline only of the proof of~\eqref{eq:capacity_high_snr_full_rank} and, furthermore, restrict ourselves to \iid channels, i.e., $\corrmat=\matI_{\blocklength}$. 
We comment on the general case at the end of the section. 
%

First, we note that $\corrmat=\matI_{\blocklength}$ implies that  the channel is memoryless, and, hence, capacity is achieved by \iid inputs.
As a consequence, 
\begin{equation}\label{eq:memoryless_case}
	\sup_{\inpprobmeas} \mi(\inpvec;\outvec)= \blocklength\sup_{\altinpprobmeas}\mi(\inp;\outp).
\end{equation}
Here, $\outp=\ch\inp+\wgn$ with $\ch,\wgn\distas\jpg(0,1)$ and  the supremum is over the distributions~$\altinpprobmeas$ on \inp that satisfy the average-power constraint $\Ex{\altinpprobmeas}{\abs{\inp}^2}\leq \snr$.
The capacity of the memoryless channel $\outp=\ch\inp+\wgn$ was first proven to grow double-logarithmically in SNR in~\cite{taricco97-07a}.
This result was then extended in~\cite[Thm.~4.2]{lapidoth03-10a}  to multiple-antenna channels with general stationary ergodic fading distribution (of finite differential entropy rate) and general noise distributions.
The proof we provide here is based on the duality technique  and is particularly simple, as it exploits the Gaussianity of the fading distribution.   
More specifically, we use duality with the density of $\outprobmeas$ given in~\eqref{eq:generic_output_distribution}, with  $\blocklength=1$ and $\alpha=[1+\log(1+\snr)]^{-1}$.
The choice of $\alpha$ might appear unmotivated, and in fact, differently from the previous section, it is hard to find an intuitive explanation for this choice, besides the fact that it simplifies the proof.       
Consider an arbitrary \altinpprobmeas satisfying  $\Ex{\altinpprobmeas}{\abs{\inp}^2}\leq \snr$;
using~\eqref{eq:duality_ub} and~\eqref{eq:generic_output_distribution}, we obtain the following upper bound on $\mi(\inp;\outp)$:
\begin{align}\label{eq:duality_bound_full_rank}
	\mi(\inp;\outp) \leq&  \log \pi + \alpha \log(1+\snr)  -\alpha\log \alpha  + \log \Gamma(\alpha)  \notag \\
	& +(1-\alpha) \Ex{\altinpprobmeas\chtran}{\log{\abs{\outp}^2}} +\alpha \frac{\Ex{\altinpprobmeas\chtran}{\abs{\outp}^2}}{1+\snr}-\difent(\outp\given \inp)  \notag \\
   %
   	%
   	%
   	%
		%
	 		\leq& \log \pi  +\alpha[1+\log(1+\snr)]-\alpha\log \alpha+ \log \Gamma(\alpha) \notag\\
	 			& + \Ex{\altinpprobmeas\chtran}{\log{\abs{\outp}^2}} -\difent(\outp\given \inp).                     %
\end{align}
The last step follows because $\Ex{\altinpprobmeas\chtran}{\abs{\outp}^2}\leq 1+\snr$ and $\alpha<1$, by assumption, so that
\begin{equation*}	
	(1-\alpha) \Ex{\altinpprobmeas\chtran}{\log{\abs{\outp}^2}} \leq    \Ex{\altinpprobmeas\chtran}{\log{\abs{\outp}^2}}.
\end{equation*}
We continue by establishing that    
\begin{equation*}	
	\Ex{\altinpprobmeas\chtran}{\log{\abs{\outp}^2}} -\difent(\outp\given \inp)= -\gamma-\log \pi -1.
\end{equation*}
This identity follows because 
\begin{equation*}	
	       \difent(\outp\given \inp)=\Ex{}{\log(1+\abs{\inp}^2)}+\log (\pi e)
\end{equation*}
and because, given \inp, the random variable $\abs{\outp}^2$ is distributed as $(1+\abs{\inp}^2)\abs{\altwgn}^2$, where $\altwgn\distas \jpg(0,1)$, so that
\begin{align*}
	    \Ex{\altinpprobmeas\chtran}{\log{\abs{\outp}^2}} &= \Ex{\inp}{\Ex{\ch,\wgn}{\log\abs{\outp}^2 \big| \inp}}   \\  
	&= \Ex{\inp}{\Ex{\altwgn}{\log\bigl[(1+\abs{\inp}^2)\abs{\altwgn}^2\bigr] \,\big| \inp}}\\       
	&=   \Ex{\inp}{\log(1+\abs{\inp}^2)} +\underbrace{\Ex{\altwgn}{\log \abs{\altwgn}^2}}_{-\gamma}.
\end{align*}
%
%
%
%
Finally, since $\alpha[1+\log(1+\snr)] =1$
and since~\cite[Eq. (337)]{lapidoth03-10a}
\begin{equation*}	
	\log \Gamma(\alpha) -\alpha \log \alpha +\log \alpha=\landauo(1), \quad \snr \to \infty,
\end{equation*}
we get
\begin{align}
	\mi(\inp;\outp) \leq& \log \pi  +\underbrace{\alpha[1+\log(1+\snr)]}_{1}    \notag \\
	    &+\underbrace{\log \Gamma(\alpha)-\alpha\log \alpha+\log\alpha}_{\landauo(1),\,\,\,\snr \to \infty} -\log \alpha \notag\\
		& + \underbrace{\Ex{\altinpprobmeas\chtran}{\log{\abs{\outp}^2}} -\difent(\outp\given \inp) }_{-\euler-\log \pi -1}        \\
		\leq&
	 -\log \alpha -\gamma +o(1), \quad \snr \to \infty  \notag\\
					=& \log\log \snr -\gamma + o(1), \quad \snr \to \infty\label{eq:ub_memoryless}.
\end{align}
This upper bound, which suffices to conclude that capacity grows at most double-logarithmically in \snr,  can   actually be  tightened.
A more careful choice of the output distribution  makes the term $\alpha[1+\log(1+\snr)]$ vanish as $\snr \to \infty$, so that $-\gamma$ in~\eqref{eq:ub_memoryless} gets replaced by $-\gamma-1$   (see~\cite[App.~VII]{lapidoth03-10a} for details).
This modified upper bound is  tight in the sense that one can find a capacity lower bound that matches it up to a $\landauo(1)$ term (see~\cite[Thm.~4.16]{lapidoth03-10a}).
To summarize, for  $\corrmat=\matI_{\blocklength}$, we have that
\begin{equation}\label{eq:capacity_iid_summary}
	\capacity(\snr)=\log\log\snr -\gamma-1+\landauo(1), \quad \snr \to \infty.
\end{equation}
We conclude by noting that the double-logarithmic growth of capacity in SNR holds for every full-rank channel covariance matrix \corrmat.
Correlation among the channel entries, however, results in a different constant term in~\eqref{eq:capacity_iid_summary}.
More specifically, the final result in~\eqref{eq:capacity_high_snr_full_rank} follows from~\cite[Lem.~4.5]{lapidoth03-10a} and from an adaptation of~\cite[Thm.~4.41]{lapidoth03-10a} to the block-fading setup considered here. 
\section{Open Problems} 
\label{sec:open_problems} 
Duality is the main tool we used to establish the novel capacity expansion~\eqref{eq:capacity_high_snr_full_rank} for the full-rank case  and to provide an alternative, simple proof of~\eqref{eq:known_result_rank_1} for the rank-1 case  (i.e., the piecewise-constant block-fading channel model). 
For the latter case, in particular, we showed how the geometry of the communication problem at hand can be used to find an output distribution that yields an asymptotically tight upper bound.
Finding a $\landauo(1)$-accurate capacity characterization when $1<\rankcorr<\blocklength$ is an interesting open problem.

Throughout the paper, we focused exclusively on the single-antenna setup.
In the multiple-antenna case, not even a pre-log characterization is available when $1<\rankcorr<\blocklength$. 
A pre-log lower bound for the single-input multiple-output (SIMO) case has been  obtained recently in~\cite{morgenshtern10-06a}.  
Surprisingly, the bound in~\cite{morgenshtern10-06a} implies that the SIMO pre-log can be larger than the pre-log in the single-input single-output case.
Establishing whether this bound is tight is an open problem.
For two further channel models of practical interest, namely, the single-antenna frequency-selective and doubly-selective fading channel the state of affairs is similar: pre-log lower bounds have been reported in~\cite{vikalo04-09a} and~\cite{kannu10-06a}, respectively.
Establishing whether these bounds are tight is again an open problem.

%
%


\bibliographystyle{IEEEtran}
\bibliography{IEEEabrv,publishers,confs-jrnls,giubib}

\end{document}

%% file: vmr-symbols-vecbold.tex
\usepackage{amssymb}
\usepackage{amsfonts}
\usepackage{mathrsfs}
\usepackage{xspace}
\usepackage{bm}
\usepackage{upgreek}

\newcommand{\safemath}[2]{\newcommand{#1}{\ensuremath{#2}\xspace}}

\safemath{\bma}{\mathbf{a}}
\safemath{\bmb}{\mathbf{b}}
\safemath{\bmc}{\mathbf{c}}
\safemath{\bmd}{\mathbf{d}}
\safemath{\bme}{\mathbf{e}}
\safemath{\bmf}{\mathbf{f}}
\safemath{\bmg}{\mathbf{g}}
\safemath{\bmh}{\mathbf{h}}
\safemath{\bmi}{\mathbf{i}}
\safemath{\bmj}{\mathbf{j}}
\safemath{\bmk}{\mathbf{k}}
\safemath{\bml}{\mathbf{l}}
\safemath{\bmm}{\mathbf{m}}
\safemath{\bmn}{\mathbf{n}}
\safemath{\bmo}{\mathbf{o}}
\safemath{\bmp}{\mathbf{p}}
\safemath{\bmq}{\mathbf{q}}
\safemath{\bmr}{\mathbf{r}}
\safemath{\bms}{\mathbf{s}}
\safemath{\bmt}{\mathbf{t}}
\safemath{\bmu}{\mathbf{u}}
\safemath{\bmv}{\mathbf{v}}
\safemath{\bmw}{\mathbf{w}}
\safemath{\bmx}{\mathbf{x}}
\safemath{\bmy}{\mathbf{y}}
\safemath{\bmz}{\mathbf{z}}
\safemath{\bmzero}{\mathbf{0}}
\safemath{\bmone}{\mathbf{1}}

\bmdefine{\biad}{a}
\bmdefine{\bibd}{b}
\bmdefine{\bicd}{c}
\bmdefine{\bidd}{d}
\bmdefine{\bied}{e}
\bmdefine{\bifd}{f}
\bmdefine{\bigd}{g}
\bmdefine{\bihd}{h}
\bmdefine{\biid}{i}
\bmdefine{\bijd}{j}
\bmdefine{\bikd}{k}
\bmdefine{\bild}{l}
\bmdefine{\bimd}{m}
\bmdefine{\bind}{n}
\bmdefine{\biod}{o}
\bmdefine{\bipd}{p}
\bmdefine{\biqd}{q}
\bmdefine{\bird}{r}
\bmdefine{\bisd}{s}
\bmdefine{\bitd}{t}
\bmdefine{\biud}{u}
\bmdefine{\bivd}{v}
\bmdefine{\biwd}{w}
\bmdefine{\bixd}{x}
\bmdefine{\biyd}{y}
\bmdefine{\bizd}{z}

\bmdefine{\bixid}{\xi}
\bmdefine{\bilambdad}{\lambda}
\bmdefine{\bimud}{\mu}
\bmdefine{\bithetad}{\theta}
\bmdefine{\biphid}{\phi}

\safemath{\bmia}{\biad}
\safemath{\bmib}{\bibd}
\safemath{\bmic}{\bicd}
\safemath{\bmid}{\bidd}
\safemath{\bmie}{\bied}
\safemath{\bmif}{\bifd}
\safemath{\bmig}{\bigd}
\safemath{\bmih}{\bihd}
\safemath{\bmii}{\biid}
\safemath{\bmij}{\bijd}
\safemath{\bmik}{\bikd}
\safemath{\bmil}{\bild}
\safemath{\bmim}{\bimd}
\safemath{\bmin}{\bind}
\safemath{\bmio}{\biod}
\safemath{\bmip}{\bipd}
\safemath{\bmiq}{\biqd}
\safemath{\bmir}{\bird}
\safemath{\bmis}{\bisd}
\safemath{\bmit}{\bitd}
\safemath{\bmiu}{\biud}
\safemath{\bmiv}{\bivd}
\safemath{\bmiw}{\biwd}
\safemath{\bmix}{\bixd}
\safemath{\bmiy}{\biyd}
\safemath{\bmiz}{\bizd}

\safemath{\bmxi}{\bixid}
\safemath{\bmlambda}{\bilambdad}
\safemath{\bmmu}{\bimud}
\safemath{\bmtheta}{\bithetad}
\safemath{\bmphi}{\biphid}

\safemath{\bA}{\mathbf{A}}
\safemath{\bB}{\mathbf{B}}
\safemath{\bC}{\mathbf{C}}
\safemath{\bD}{\mathbf{D}}
\safemath{\bE}{\mathbf{E}}
\safemath{\bF}{\mathbf{F}}
\safemath{\bG}{\mathbf{G}}
\safemath{\bH}{\mathbf{H}}
\safemath{\bI}{\mathbf{I}}
\safemath{\bJ}{\mathbf{J}}
\safemath{\bK}{\mathbf{K}}
\safemath{\bL}{\mathbf{L}}
\safemath{\bM}{\mathbf{M}}
\safemath{\bN}{\mathbf{N}}
\safemath{\bO}{\mathbf{O}}
\safemath{\bP}{\mathbf{P}}
\safemath{\bQ}{\mathbf{Q}}
\safemath{\bR}{\mathbf{R}}
\safemath{\bS}{\mathbf{S}}
\safemath{\bT}{\mathbf{T}}
\safemath{\bU}{\mathbf{U}}
\safemath{\bV}{\mathbf{V}}
\safemath{\bW}{\mathbf{W}}
\safemath{\bX}{\mathbf{X}}
\safemath{\bY}{\mathbf{Y}}
\safemath{\bZ}{\mathbf{Z}}

\safemath{\bZero}{\mathbf{0}}
\safemath{\bOne}{\mathbf{1}}
\safemath{\bDelta}{\mathbf{\Delta}}
\safemath{\bLambda}{\mathbf{\UpLambda}}
\safemath{\bPhi}{\mathbf{\Upphi}}
\safemath{\bSigma}{\mathbf{\Upsigma}}
\safemath{\bOmega}{\mathbf{\Upomega}}
\safemath{\bTheta}{\mathbf{\Uptheta}}

\bmdefine{\biAd}{A}
\bmdefine{\biBd}{B}
\bmdefine{\biCd}{C}
\bmdefine{\biDd}{D}
\bmdefine{\biEd}{E}
\bmdefine{\biFd}{F}
\bmdefine{\biGd}{G}
\bmdefine{\biHd}{H}
\bmdefine{\biId}{I}
\bmdefine{\biJd}{J}
\bmdefine{\biKd}{K}
\bmdefine{\biLd}{L}
\bmdefine{\biMd}{M}
\bmdefine{\biOd}{N}
\bmdefine{\biPd}{O}
\bmdefine{\biQd}{P}
\bmdefine{\biRd}{R}
\bmdefine{\biSd}{S}
\bmdefine{\biTd}{T}
\bmdefine{\biUd}{U}
\bmdefine{\biVd}{V}
\bmdefine{\biWd}{W}
\bmdefine{\biXd}{X}
\bmdefine{\biYd}{Y}
\bmdefine{\biZd}{Z}

\bmdefine{\biDelta}{\Delta}
\bmdefine{\biLambda}{\Lambda}
\bmdefine{\biPhi}{\Phi}
\bmdefine{\biSigma}{\Sigma}
\bmdefine{\biOmega}{\Omega}
\bmdefine{\biTheta}{\Theta}

\safemath{\bimA}{\biAd}
\safemath{\bimB}{\biBd}
\safemath{\bimC}{\biCd}
\safemath{\bimD}{\biDd}
\safemath{\bimE}{\biEd}
\safemath{\bimF}{\biFd}
\safemath{\bimG}{\biGd}
\safemath{\bimH}{\biHd}
\safemath{\bimI}{\biId}
\safemath{\bimJ}{\biJd}
\safemath{\bimK}{\biKd}
\safemath{\bimL}{\biLd}
\safemath{\bimM}{\biMd}
\safemath{\bimN}{\biNd}
\safemath{\bimO}{\biOd}
\safemath{\bimP}{\biPd}
\safemath{\bimQ}{\biQd}
\safemath{\bimR}{\biRd}
\safemath{\bimS}{\biSd}
\safemath{\bimT}{\biTd}
\safemath{\bimU}{\biUd}
\safemath{\bimV}{\biVd}
\safemath{\bimW}{\biWd}
\safemath{\bimX}{\biXd}
\safemath{\bimY}{\biYd}
\safemath{\bimZ}{\biZd}

\safemath{\bimDelta}{\biDelta}
\safemath{\bimLambda}{\biLambda}
\safemath{\bimPhi}{\biPhi}
\safemath{\bimSigma}{\biSigma}
\safemath{\bimOmega}{\biOmega}
\safemath{\bimTheta}{\biTheta}

\safemath{\setA}{\mathcal{A}}
\safemath{\setB}{\mathcal{B}}
\safemath{\setC}{\mathcal{C}}
\safemath{\setD}{\mathcal{D}}
\safemath{\setE}{\mathcal{E}}
\safemath{\setF}{\mathcal{F}}
\safemath{\setG}{\mathcal{G}}
\safemath{\setH}{\mathcal{H}}
\safemath{\setI}{\mathcal{I}}
\safemath{\setJ}{\mathcal{J}}
\safemath{\setK}{\mathcal{K}}
\safemath{\setL}{\mathcal{L}}
\safemath{\setM}{\mathcal{M}}
\safemath{\setN}{\mathcal{N}}
\safemath{\setO}{\mathcal{O}}
\safemath{\setP}{\mathcal{P}}
\safemath{\setQ}{\mathcal{Q}}
\safemath{\setR}{\mathcal{R}}
\safemath{\setS}{\mathcal{S}}
\safemath{\setT}{\mathcal{T}}
\safemath{\setU}{\mathcal{U}}
\safemath{\setV}{\mathcal{V}}
\safemath{\setW}{\mathcal{W}}
\safemath{\setX}{\mathcal{X}}
\safemath{\setY}{\mathcal{Y}}
\safemath{\setZ}{\mathcal{Z}}
\safemath{\emptySet}{\varnothing}

\safemath{\colA}{\mathscr{A}}
\safemath{\colB}{\mathscr{B}}
\safemath{\colC}{\mathscr{C}}
\safemath{\colD}{\mathscr{D}}
\safemath{\colE}{\mathscr{E}}
\safemath{\colF}{\mathscr{F}}
\safemath{\colG}{\mathscr{G}}
\safemath{\colH}{\mathscr{H}}
\safemath{\colI}{\mathscr{I}}
\safemath{\colJ}{\mathscr{J}}
\safemath{\colK}{\mathscr{K}}
\safemath{\colL}{\mathscr{L}}
\safemath{\colM}{\mathscr{M}}
\safemath{\colN}{\mathscr{N}}
\safemath{\colO}{\mathscr{O}}
\safemath{\colP}{\mathscr{P}}
\safemath{\colQ}{\mathscr{Q}}
\safemath{\colR}{\mathscr{R}}
\safemath{\colS}{\mathscr{S}}
\safemath{\colT}{\mathscr{T}}
\safemath{\colU}{\mathscr{U}}
\safemath{\colV}{\mathscr{V}}
\safemath{\colW}{\mathscr{W}}
\safemath{\colX}{\mathscr{X}}
\safemath{\colY}{\mathscr{Y}}
\safemath{\colZ}{\mathscr{Z}}

\safemath{\opA}{\mathbb{A}}
\safemath{\opB}{\mathbb{B}}
\safemath{\opC}{\mathbb{C}}
\safemath{\opD}{\mathbb{D}}
\safemath{\opE}{\mathbb{E}}
\safemath{\opF}{\mathbb{F}}
\safemath{\opG}{\mathbb{G}}
\safemath{\opH}{\mathbb{H}}
\safemath{\opI}{\mathbb{I}}
\safemath{\opJ}{\mathbb{J}}
\safemath{\opK}{\mathbb{K}}
\safemath{\opL}{\mathbb{L}}
\safemath{\opM}{\mathbb{M}}
\safemath{\opN}{\mathbb{N}}
\safemath{\opO}{\mathbb{O}}
\safemath{\opP}{\mathbb{P}}
\safemath{\opQ}{\mathbb{Q}}
\safemath{\opR}{\mathbb{R}}
\safemath{\opS}{\mathbb{S}}
\safemath{\opT}{\mathbb{T}}
\safemath{\opU}{\mathbb{U}}
\safemath{\opV}{\mathbb{V}}
\safemath{\opW}{\mathbb{W}}
\safemath{\opX}{\mathbb{X}}
\safemath{\opY}{\mathbb{Y}}
\safemath{\opZ}{\mathbb{Z}}
\safemath{\opZero}{\mathbb{O}}
\safemath{\identityop}{\opI}

\safemath{\veca}{\bma}
\safemath{\vecb}{\bmb}
\safemath{\vecc}{\bmc}
\safemath{\vecd}{\bmd}
\safemath{\vece}{\bme}
\safemath{\vecf}{\bmf}
\safemath{\vecg}{\bmg}
\safemath{\vech}{\bmh}
\safemath{\veci}{\bmi}
\safemath{\vecj}{\bmj}
\safemath{\veck}{\bmk}
\safemath{\vecl}{\bml}
\safemath{\vecm}{\bmm}
\safemath{\vecn}{\bmn}
\safemath{\veco}{\bmo}
\safemath{\vecp}{\bmp}
\safemath{\vecq}{\bmq}
\safemath{\vecr}{\bmr}
\safemath{\vecs}{\bms}
\safemath{\vect}{\bmt}
\safemath{\vecu}{\bmu}
\safemath{\vecv}{\bmv}
\safemath{\vecw}{\bmw}
\safemath{\vecx}{\bmx}
\safemath{\vecy}{\bmy}
\safemath{\vecz}{\bmz}

\safemath{\veczero}{\bmzero}
\safemath{\vecone}{\bmone}
\safemath{\vecxi}{\bmxi}
\safemath{\veclambda}{\bmlambda}
\safemath{\vecmu}{\bmmu}
\safemath{\vectheta}{\bmtheta}
\safemath{\vecphi}{\bmphi}

\safemath{\matA}{\bA}
\safemath{\matB}{\bB}
\safemath{\matC}{\bC}
\safemath{\matD}{\bD}
\safemath{\matE}{\bE}
\safemath{\matF}{\bF}
\safemath{\matG}{\bG}
\safemath{\matH}{\bH}
\safemath{\matI}{\bI}
\safemath{\matJ}{\bJ}
\safemath{\matK}{\bK}
\safemath{\matL}{\bL}
\safemath{\matM}{\bM}
\safemath{\matN}{\bN}
\safemath{\matO}{\bO}
\safemath{\matP}{\bP}
\safemath{\matQ}{\bQ}
\safemath{\matR}{\bR}
\safemath{\matS}{\bS}
\safemath{\matT}{\bT}
\safemath{\matU}{\bU}
\safemath{\matV}{\bV}
\safemath{\matW}{\bW}
\safemath{\matX}{\bX}
\safemath{\matY}{\bY}
\safemath{\matZ}{\bZ}
\safemath{\matzero}{\bmzero}

\safemath{\matDelta}{\bDelta}
\safemath{\matLambda}{\bLambda}
\safemath{\matPhi}{\bPhi}
\safemath{\matSigma}{\bSigma}
\safemath{\matOmega}{\bOmega}
\safemath{\matTheta}{\bTheta}

\safemath{\matidentity}{\matI}
\safemath{\matone}{\matO}

\safemath{\rnda}{A}
\safemath{\rndb}{B}
\safemath{\rndc}{C}
\safemath{\rndd}{D}
\safemath{\rnde}{E}
\safemath{\rndf}{F}
\safemath{\rndg}{G}
\safemath{\rndh}{H}
\safemath{\rndi}{I}
\safemath{\rndj}{J}
\safemath{\rndk}{K}
\safemath{\rndl}{L}
\safemath{\rndm}{M}
\safemath{\rndn}{N}
\safemath{\rndo}{O}
\safemath{\rndp}{P}
\safemath{\rndq}{Q}
\safemath{\rndr}{R}
\safemath{\rnds}{S}
\safemath{\rndt}{T}
\safemath{\rndu}{U}
\safemath{\rndv}{V}
\safemath{\rndw}{W}
\safemath{\rndx}{X}
\safemath{\rndy}{Y}
\safemath{\rndz}{Z}

\safemath{\rveca}{\bimA}
\safemath{\rvecb}{\bimB}
\safemath{\rvecc}{\bimC}
\safemath{\rvecd}{\bimD}
\safemath{\rvece}{\bimE}
\safemath{\rvecf}{\bimF}
\safemath{\rvecg}{\bimG}
\safemath{\rvech}{\bimH}
\safemath{\rveci}{\bimI}
\safemath{\rvecj}{\bimJ}
\safemath{\rveck}{\bimK}
\safemath{\rvecl}{\bimL}
\safemath{\rvecm}{\bimM}
\safemath{\rvecn}{\bimN}
\safemath{\rveco}{\bomO}
\safemath{\rvecp}{\bimP}
\safemath{\rvecq}{\bimQ}
\safemath{\rvecr}{\bimR}
\safemath{\rvecs}{\bimS}
\safemath{\rvect}{\bimT}
\safemath{\rvecu}{\bimU}
\safemath{\rvecv}{\bimV}
\safemath{\rvecw}{\bimW}
\safemath{\rvecx}{\bimX}
\safemath{\rvecy}{\bimY}
\safemath{\rvecz}{\bimZ}

\safemath{\rvecxi}{\bmxi}
\safemath{\rveclambda}{\bmlambda}
\safemath{\rvecmu}{\bmmu}
\safemath{\rvectheta}{\bmtheta}
\safemath{\rvecphi}{\bmphi}

\safemath{\rmatA}{\bimA}
\safemath{\rmatB}{\bimB}
\safemath{\rmatC}{\bimC}
\safemath{\rmatD}{\bimD}
\safemath{\rmatE}{\bimE}
\safemath{\rmatF}{\bimF}
\safemath{\rmatG}{\bimG}
\safemath{\rmatH}{\bimH}
\safemath{\rmatI}{\bimI}
\safemath{\rmatJ}{\bimJ}
\safemath{\rmatK}{\bimK}
\safemath{\rmatL}{\bimL}
\safemath{\rmatM}{\bimM}
\safemath{\rmatN}{\bimN}
\safemath{\rmatO}{\bimO}
\safemath{\rmatP}{\bimP}
\safemath{\rmatQ}{\bimQ}
\safemath{\rmatR}{\bimR}
\safemath{\rmatS}{\bimS}
\safemath{\rmatT}{\bimT}
\safemath{\rmatU}{\bimU}
\safemath{\rmatV}{\bimV}
\safemath{\rmatW}{\bimW}
\safemath{\rmatX}{\bimX}
\safemath{\rmatY}{\bimY}
\safemath{\rmatZ}{\bimZ}

\safemath{\rmatDelta}{\bimDelta}
\safemath{\rmatLambda}{\bimLambda}
\safemath{\rmatPhi}{\bimPhi}
\safemath{\rmatSigma}{\bimSigma}
\safemath{\rmatOmega}{\bimOmega}
\safemath{\rmatTheta}{\bimTheta}

%% file: standard-macros.tex
\usepackage{amssymb}
\usepackage{amsfonts}
\usepackage{mathrsfs}
\usepackage{xspace}
\usepackage{bm}
\usepackage{fancyref}
\usepackage{textcomp}

\newenvironment{textbmatrix}{	\setlength{\arraycolsep}{2.5pt}%
								\big[\begin{matrix}}{\end{matrix}\big]%
								\raisebox{0.08ex}{\vphantom{M}}}

\def\be{\begin{equation}}
\def\ee{\end{equation}}
\def\een{\nonumber \end{equation}}
\def\mat{\begin{bmatrix}}
\def\emat{\end{bmatrix}}
\def\btm{\begin{textbmatrix}}
\def\etm{\end{textbmatrix}}

\def\ba#1\ea{\begin{align}#1\end{align}}
\def\bas#1\eas{\begin{align*}#1\end{align*}}
\def\bs#1\es{\begin{split}#1\end{split}} 
\def\bg#1\eg{\begin{gather}#1\end{gather}}
\def\bml#1\eml{\begin{multline}#1\end{multline}}
\def\bi#1\ei{\begin{itemize}#1\end{itemize}}

\newcommand{\subtext}[1]{\text{\fontfamily{cmr}\fontshape{n}\fontseries{m}\selectfont{}#1}}
\newcommand{\lefto}{\mathopen{}\left}

\newcommand{\sub}[1]{\ensuremath{_{\subtext{#1}}}} 

\DeclareMathOperator{\Prob}{\opP}			
\DeclareMathOperator{\Exop}{\opE}			

\DeclareMathOperator{\landauo}{\mathit{o}}
\DeclareMathOperator{\landauO}{\mathcal{O}}

\newcommand{\Ex}[2]{\ensuremath{\Exop_{#1}\lefto[#2\right]}} 	
\newcommand{\abs}[1]{\left\lvert#1\right\rvert}		

\newcommand{\vecnorm}[1]{\lVert#1\rVert}		
\newcommand{\tp}[1]{\ensuremath{#1^{T}}} 		
\newcommand{\herm}[1]{\ensuremath{#1^{H}}} 	

\safemath{\dirac}{\delta}					
\safemath{\krond}{\dirac}					
\newcommand{\allo}[2]{\ensuremath{#1=1,2,\ldots,#2}}

\safemath{\upto}{\uparrow}
\safemath{\downto}{\downarrow}
\safemath{\iu}{j}							
\safemath{\ev}{\lambda}						
\safemath{\hilseqspace}{l^{2}}				
\newcommand{\banachfunspace}[1]{\setL^{#1}}	
\safemath{\hilfunspace}{\banachfunspace{2}}	

\safemath{\snr}{\rho} 				
\safemath{\No}{N_0}							
\safemath{\Es}{E_s}							
\safemath{\Eb}{E_b}							
\safemath{\EbNo}{\frac{\Eb}{\No}}
\safemath{\EsNo}{\frac{\Es}{\No}}

\DeclareMathOperator{\CHop}{\ensuremath{\opH}} 
\safemath{\tvir}{\rndh_{\CHop}}				
\safemath{\tvtf}{\rndl_{\CHop}}				
\safemath{\spf}{\rnds_{\CHop}}				
\safemath{\bff}{H_{\CHop}}					

\safemath{\ircf}{r_{h}}						
\safemath{\tftvcf}{r_{s}}					
\safemath{\tfcf}{r_{l}}						
\safemath{\bfcf}{r_{H}}						

\safemath{\tcorr}{c_h}						
\safemath{\scf}{c_{s}}						
\safemath{\tfcorr}{c_{l}}					
\safemath{\fcorr}{c_{H}}						

\safemath{\mi}{I}							
\safemath{\capacity}{C}						

\newcommand{\iid}{i.i.d.\@\xspace}

\safemath{\normal}{\mathcal{N}}			
\safemath{\jpg}{\mathcal{CN}}			
\safemath{\mchain}{\leftrightarrow}		
\newcommand{\given}{\,\vert\,}				

\safemath{\dB}{\,\mathrm{dB}}
\safemath{\dBm}{\,\mathrm{dBm}}
\safemath{\Hz}{\,\mathrm{Hz}}
\safemath{\kHz}{\,\mathrm{kHz}}
\safemath{\MHz}{\,\mathrm{MHz}}
\safemath{\GHz}{\,\mathrm{GHz}}
\safemath{\s}{\,\mathrm{s}}
\safemath{\ms}{\,\mathrm{ms}}
\safemath{\mus}{\,\mathrm{\text{\textmu}s}}
\safemath{\ns}{\,\mathrm{ns}}
\safemath{\ps}{\,\mathrm{ps}}
\safemath{\meter}{\,\mathrm{m}}
\safemath{\mm}{\,\mathrm{mm}}
\safemath{\cm}{\,\mathrm{cm}}
\safemath{\m}{\,\mathrm{m}}
\safemath{\W}{\,\mathrm{W}}
\safemath{\mW}{\, \mathrm{mW}}
\safemath{\J}{\,\mathrm{J}}
\safemath{\K}{\,\mathrm{K}}
\safemath{\bit}{\,\mathrm{bit}}
\safemath{\nat}{\,\mathrm{nat}}

\safemath{\define}{\triangleq}			

\newcommand{\sothat}{\,:\,}				
\safemath{\equivalent}{\sim}
\safemath{\distas}{\sim}					
\safemath{\sdiff}{\Delta}				

\safemath{\reals}{\mathbb{R}}
\safemath{\positivereals}{\reals_{+}}
\safemath{\integers}{\mathbb{Z}}
\safemath{\posint}{\integers_{+}}
\safemath{\naturals}{\mathbb{N}}
\safemath{\posnaturals}{\naturals_{+}}
\safemath{\complexset}{\mathbb{C}}
\safemath{\rationals}{\mathbb{Q}}

\newcommand*{\fancyrefapplabelprefix}{app}		
\newcommand*{\fancyrefthmlabelprefix}{thm}		
\newcommand*{\fancyreflemlabelprefix}{lem}		
\newcommand*{\fancyrefcorlabelprefix}{cor}		
\newcommand*{\fancyrefdeflabelprefix}{def}		
\newcommand*{\fancyrefproplabelprefix}{prop}		
\newcommand*{\fancyrefexelabelprefix}{exe}		
\frefformat{vario}{\fancyrefseclabelprefix}{Section~#1}
\frefformat{vario}{\fancyrefthmlabelprefix}{Theorem~#1}
\frefformat{vario}{\fancyreflemlabelprefix}{Lemma~#1}
\frefformat{vario}{\fancyrefcorlabelprefix}{Corollary~#1}
\frefformat{vario}{\fancyrefdeflabelprefix}{Definition~#1}
\frefformat{vario}{\fancyreffiglabelprefix}{Fig.~#1}
\frefformat{vario}{\fancyrefapplabelprefix}{Appendix~#1}
\frefformat{vario}{\fancyrefeqlabelprefix}{(#1)}
\frefformat{vario}{\fancyrefproplabelprefix}{Property~#1}
\frefformat{vario}{\fancyrefexelabelprefix}{Exercise~#1}

%% file: defs.tex
\safemath{\outvec}{\vecy} 
\safemath{\outp}{y}		  
\safemath{\inpvec}{\vecx} 
\safemath{\inp}{x}		  
\safemath{\wgnvec}{\vecw} 
\safemath{\wgn}{w}		  
\safemath{\altwgn}{z}	  
\safemath{\ch}{s}		 	  
\safemath{\chvec}{\vecs}	
\safemath{\outvecnn}{\vecr}	
\safemath{\chgenvec}{\vech}	
\safemath{\chgenvecdir}{\widehat{\chgenvec}}	
\safemath{\chgen}{h}		
\safemath{\inpprobmeas}{\mathsf{Q}}
\safemath{\altinpprobmeas}{\widetilde{\inpprobmeas}}
\safemath{\outprobmeas}{\mathsf{R}}
\safemath{\chtran}{\mathsf{W}}
\safemath{\dtime}{n}			
\safemath{\chmat}{\matS}		

\safemath{\corrmat}{\matR}	
\safemath{\rankcorr}{Q}		
\safemath{\inpmin}{\inp\sub{min}}

\safemath{\sqrtmat}{\matP}	
\safemath{\sqrtmatrow}{\vecp} 
\newcommand{\id}[1]{\vecu_{#1}}
\safemath{\blocklength}{N}	
\safemath{\prelog}{\chi}	
\safemath{\euler}{\gamma}	
\safemath{\difent}{\mathrm{h}}		
\safemath{\difentsphere}{\difent_{\text{sphere}}} 

\safemath{\numtx}{M\sub{T}} 
\safemath{\numrx}{M\sub{R}} 

\safemath{\eig}{\lambda}	
\safemath{\eigmin}{\eig\sub{min}}	

\newcommand{\relent}[2]{D(#1 \| #2)} 

\newtheorem{thm}{Theorem}
\newtheorem{lem}[thm]{Lemma}

\safemath{\capacityas}{\capacity_{\text{as}}}		
\safemath{\capacityconstr}{\capacity_{\setK}}	
\safemath{\asmap}{d}								

\safemath{\inpvecone}{\inpvec_{\rankcorr}}
\safemath{\outvecone}{\outvec_{\rankcorr}}

\safemath{\inpvectwo}{\inpvec_{\blocklength-\rankcorr}}
\safemath{\outvectwo}{\outvec_{\blocklength-\rankcorr}}
\DeclareMathOperator{\diag}{diag}
\safemath{\outpdf}{\mathsf{r}}
\safemath{\wpone}{w.p.1}

%% file: db_aeu11_arxiv.bbl
\begin{thebibliography}{10}
\providecommand{\url}[1]{#1}
\csname url@samestyle\endcsname
\providecommand{\newblock}{\relax}
\providecommand{\bibinfo}[2]{#2}
\providecommand{\BIBentrySTDinterwordspacing}{\spaceskip=0pt\relax}
\providecommand{\BIBentryALTinterwordstretchfactor}{4}
\providecommand{\BIBentryALTinterwordspacing}{\spaceskip=\fontdimen2\font plus
\BIBentryALTinterwordstretchfactor\fontdimen3\font minus
  \fontdimen4\font\relax}
\providecommand{\BIBforeignlanguage}[2]{{%
\expandafter\ifx\csname l@#1\endcsname\relax
\typeout{** WARNING: IEEEtran.bst: No hyphenation pattern has been}%
\typeout{** loaded for the language `#1'. Using the pattern for}%
\typeout{** the default language instead.}%
\else
\language=\csname l@#1\endcsname
\fi
#2}}
\providecommand{\BIBdecl}{\relax}
\BIBdecl

\bibitem{tong04-11a}
L.~Tong, B.~M. Sadler, and M.~Dong, ``Pilot-assisted wireless transmissions,''
  \emph{{IEEE} Signal Process. Mag.}, vol.~21, no.~6, pp. 12--25, Nov. 2004.

\bibitem{cover06-a}
T.~M. Cover and J.~A. Thomas, \emph{Elements of Information Theory},
  2nd~ed.\hskip 1em plus 0.5em minus 0.4em\relax New York, NY, U.S.A.: Wiley,
  2006.

\bibitem{abou-faycal01-05a}
I.~C. Abou-Faycal, M.~D. Trott, and S.~Shamai~(Shitz), ``The capacity of
  discrete-time memoryless {Rayleigh}-fading channels,'' \emph{{IEEE} Trans.
  Inf. Theory}, vol.~47, no.~4, pp. 1290--1301, May 2001.

\bibitem{lapidoth03-10a}
A.~Lapidoth and S.~M. Moser, ``Capacity bounds via duality with applications to
  multiple-antenna systems on flat-fading channels,'' \emph{{IEEE} Trans. Inf.
  Theory}, vol.~49, no.~10, pp. 2426--2467, Oct. 2003.

\bibitem{lapidoth05-07a}
A.~Lapidoth, ``On the asymptotic capacity of stationary {Gaussian} fading
  channels,'' \emph{{IEEE} Trans. Inf. Theory}, vol.~51, no.~2, pp. 437--446,
  Feb. 2005.

\bibitem{durisi11-03b}
G.~Durisi, V.~I. Morgenshtern, H.~B{\"o}lcskei, U.~G. Schuster, and
  S.~Shamai~(Shitz), ``Information theory of underspread {WSSUS} channels,'' in
  \emph{Wireless Communications over Rapidly Time-Varying Channels},
  F.~Hlawatsch and G.~Matz, Eds.\hskip 1em plus 0.5em minus 0.4em\relax
  Academic Press, Mar. 2011, ch.~2, pp. 65--115.

\bibitem{liang04-09a}
Y.~Liang and V.~V. Veeravalli, ``Capacity of noncoherent time-selective
  {Rayleigh}-fading channels,'' \emph{{IEEE} Trans. Inf. Theory}, vol.~50,
  no.~12, pp. 3095--3110, Dec. 2004.

\bibitem{morgenshtern10-06a}
V.~I. Morgenshtern, G.~Durisi, and H.~B\"olcskei, ``The {SIMO} pre-log can be
  larger than the {SISO} pre-log,'' in \emph{Proc. IEEE Int. Symp. Inf. Theory
  (ISIT)}, Austin, TX, U.S.A., Jun. 2010, pp. 320--324.

\bibitem{zheng02-02a}
L.~Zheng and D.~N.~C. Tse, ``Communication on the {Grassmann} manifold: A
  geometric approach to the noncoherent multiple-antenna channel,''
  \emph{{IEEE} Trans. Inf. Theory}, vol.~48, no.~2, pp. 359--383, Feb. 2002.

\bibitem{hochwald00-03a}
B.~M. Hochwald and T.~L. Marzetta, ``Unitary space--time modulation for
  multiple-antenna communications in {Rayleigh} flat fading,'' \emph{{IEEE}
  Trans. Inf. Theory}, vol.~46, no.~2, pp. 543--564, Mar. 2000.

\bibitem{lapidoth09a}
A.~Lapidoth, \emph{A Foundation in Digital Communication}.\hskip 1em plus 0.5em
  minus 0.4em\relax Cambridge, U.K.: Cambridge Univ. Press, 2009.

\bibitem{peled80-04a}
A.~Peled and A.~Ruiz, ``Frequency domain data transmission using reduced
  computational complexity algorithms,'' in \emph{Proc. IEEE Int. Conf.
  Acoust., Speech, Signal Process. (ICASSP)}, vol.~5, Denver, CO, U.S.A., Apr.
  1980, pp. 964--967.

\bibitem{marzetta99-01a}
T.~L. Marzetta and B.~M. Hochwald, ``Capacity of a mobile multiple-antenna
  communication link in {Rayleigh} flat fading,'' \emph{{IEEE} Trans. Inf.
  Theory}, vol.~45, no.~1, pp. 139--157, Jan. 1999.

\bibitem{biglieri98-10a}
E.~Biglieri, J.~Proakis, and S.~Shamai~(Shitz), ``Fading channels:
  Information-theoretic and communications aspects,'' \emph{{IEEE} Trans. Inf.
  Theory}, vol.~44, no.~6, pp. 2619--2692, Oct. 1998.

\bibitem{topsoe67-a}
F.~Tops{\o}e, ``An information theoretical identity and a problem involving
  capacity,'' \emph{Studia Scientiarum Math. Hung.}, vol.~2, pp. 291--292,
  1967.

\bibitem{csiszar82-a}
I.~Csisz\'ar and J.~K\"orner, \emph{Information theory: Coding theorems for
  discrete memoryless systems}.\hskip 1em plus 0.5em minus 0.4em\relax Orlando,
  FL, U.S.A.: Academic Press, Inc., 1982.

\bibitem{lapidoth06-02a}
A.~Lapidoth and S.~Moser, ``The fading number of single-input multiple-output
  fading channels with memory,'' \emph{{IEEE} Trans. Inf. Theory}, vol.~52,
  no.~2, pp. 437--453, Feb. 2006.

\bibitem{boothby86-a}
W.~M. Boothby, \emph{An Introduction to Differentiable Manifolds and Riemannian
  Geometry}, 2nd~ed.\hskip 1em plus 0.5em minus 0.4em\relax Orlando, FL,
  U.S.A.: Academic Press, Inc., 1986.

\bibitem{taricco97-07a}
G.~Taricco and M.~Elia, ``Capacity of fading channel with no side
  information,'' \emph{Electron. Lett.}, vol.~33, no.~16, pp. 1368--1370, Jul.
  1997.

\bibitem{vikalo04-09a}
H.~Vikalo, B.~Hassibi, B.~M. Hochwald, and T.~Kailath, ``On the capacity of
  frequency-selective channels in training-based transmission schemes,''
  \emph{{IEEE} Trans. Signal Process.}, vol.~52, no.~9, pp. 2572--2583, Sep.
  2004.

\bibitem{kannu10-06a}
A.~P. Kannu and P.~Schniter, ``On the spectral efficiency of noncoherent doubly
  selective channels,'' \emph{{IEEE} Trans. Inf. Theory}, vol.~56, no.~6, pp.
  2829--2844, Jun. 2010.

\end{thebibliography}
